\documentclass[a4paper,UKenglish,cleveref, autoref, thm-restate]{lipics-v2021}

\usepackage{algorithmic}
\usepackage{algorithm}
\usepackage{wrapfig}
\usepackage{amsmath}
\usepackage{graphicx}

\title{Parsimonious Learning-Augmented Approximations for Dense Instances of $\mathcal{NP}$-hard Problems} 

\author{Evripidis Bampis}{Sorbonne Universit\'e, CNRS, LIP6, F-75005 Paris, France}{evripidis.bampis@lip6.fr}{}{}

\author{Bruno Escoffier}{Sorbonne Universit\'e, CNRS, LIP6, F-75005 Paris, France \and Institut Universitaire de France, Paris, France}{bruno.escoffier@lip6.fr}{}{
}

\author{Michalis Xefteris}{Sorbonne Universit\'e, CNRS, LIP6, F-75005 Paris, France}{michail.xefteris@lip6.fr}{}{
}

\authorrunning{E. Bampis, B. Escoffier and M. Xefteris} 

\Copyright{Evripidis Bampis, Bruno Escoffier, Michalis Xefteris} 



\category{} 

\relatedversion{} 

\funding{This work was partially funded by the grant ANR-19-CE48-0016 from the French National Research Agency (ANR).}

\acknowledgements{
}

\nolinenumbers 

\begin{document}
\maketitle
\begin{abstract}
    The classical work of~\cite{dense} provides a scheme that gives, for any $\epsilon>0$, a polynomial time $1-\epsilon$ approximation algorithm for dense instances of a family of $\mathcal{NP}$-hard problems, such as \textsc{Max-CUT} and \textsc{Max}-$k$-\textsc{SAT}. In this paper we extend and speed up this scheme using a logarithmic number of one-bit predictions. We propose a learning augmented framework which aims at finding fast algorithms which guarantees approximation consistency, smoothness and robustness with respect to the prediction error. We provide such algorithms, which moreover use predictions parsimoniously, for dense instances of various optimization problems.
\end{abstract}
    
\section{Introduction}
In an era marked by the widespread adoption of Machine Learning technology, ML predictors, capable of learning to predict the unknown based on (past) data, are employed to solve numerous problems daily. Due to an effort to exploit this development, there has been a trend in recent years that tries to use ML predictions in order to overcome known worst-case computational limitations. The goal is to provide algorithms that use a possibly erroneous predictor to enhance their performance when the prediction is accurate, while still providing worst case performance guarantees.

The formal framework for these learning-augmented algorithms (or algorithms with predictions) has been presented by Lykouris and Vassilvitskii in their seminal paper~\cite{lykouris}, in which they studied the caching problem. In this framework, no assumption is made about the quality of the predictor and the objective is to design learning-augmented algorithms that are \textit{consistent}, i.e., whose performance is close to the best possible performance that can be achieved when the prediction is perfect, \textit{smooth}, meaning that the quality of the  solution produced degrades smoothly with the error made in the prediction, and \textit{robust}. The robustness requires that the performance of the algorithm remains close to the one of the best worst-case algorithm even when the prediction is bad (see Section~\ref{notation} for formal definitions).

This vein of work has produced various results for online algorithms, i.e., algorithms that are not aware of the whole (future) input of the problem, including scheduling~\cite{mitzenmacher, kumar_1}, metrical task systems~\cite{antoniadis_2}, online facility location~\cite{facility} and online routing problems~\cite{bampis, evripidis_2}. More related to our work, in~\cite{svensson} they use predictions to speed up the Bellman-Ford algorithm for the shortest path problem. Furthermore, in~\cite{warmstart} they design a faster algorithm for computing matchings utilizing warm-start predicted solutions, and in~\cite{lu, bai} they speed up sorting using predictions. For clustering, Ergund et al. present an algorithm that given a prediction with error rate upper bounded by $\alpha$ achieves an approximation of $1+O(\alpha)$ in almost optimal running time~\cite{clustering}.

All the aforementioned works use the predictor without any limitations. Recently, a new line of work, which uses a small number of predictions to design learning-augmented algorithms, has emerged. Im et al. proposed an algorithm that uses a bounded number of predictions to solve the online caching problem~\cite{kumar}. In~\cite{antoniadis} they solved the paging problem utilizing a minimum amount of predictions. Similar works that present algorithms which take into account the amount of predictions used, penalizing each prediction request by some cost or given a finite budget are the works of~\cite{drygala} and~\cite{benomar}, respectively. In this paper, we follow this line of work and use a logarithmic number of one-bit predictions to solve dense instances for a family of problems that includes \textsc{Max-CUT} and \textsc{MAX}-$k$-\textsc{SAT}.

Simultaneously and independently of our work, \textsc{MaxCut} with predictions was studied in two separate papers.
Cohen-Addad et al.~\cite{cohenaddad} studied the approximability of \textsc{MaxCut} with predictions considering two different models. In their first model, they get a prediction for each vertex that is correct with probability $1/2+\epsilon$ and give a polynomial-time ($0.878+\Tilde{\Omega}(\epsilon^4)$)-approximation algorithm. In their second model, they get a correct prediction for each vertex with probability $\epsilon$ and design a ($0.858+\Omega(\epsilon)$)-approximation algorithm. Ghoshal et al.~\cite{constraint} studied \textsc{MaxCut} and \textsc{Max2-Lin} in the two aforementioned models as well.

\paragraph*{Approximation algorithms} Approximation algorithms are one standard way of dealing with $\mathcal{NP}$-hard problems as they usually run in polynomial time. An algorithm is an $\alpha$-approximation for an optimization problem iff for every instance of the problem it can find a solution within a factor of $\alpha$ of the optimum solution. 
If the problem is a maximization problem, $\alpha \le 1$ and the approximate solution is at least $\alpha$ times the optimum. 

A PTAS (Polynomial Time Approximation Scheme) takes an instance of an optimization problem and a user-defined parameter $\epsilon > 0$  and outputs a solution that is within a factor $1-\epsilon$ of being optimal (or $1+\epsilon$ for minimization problems). The running time of a PTAS is required to be polynomial in the problem size for every fixed $\epsilon>0$, but can be even super-exponential with respect to $1/\epsilon$. However, hardness results have shown that unless $\mathcal{P} = \mathcal{NP}$, problems such as vertex cover, \textsc{Max}-$3$-\textsc{SAT}, \textsc{Max-CUT} and metric TSP do not have a PTAS~\cite{sudan, papad}. Moreover, $k$-\textsc{Densest Subgraph} does not admit a PTAS under a complexity assumption~\cite{khot}. 

Despite the discouraging results, many approximation algorithms for \textsc{MAX}-$\mathcal{SNP}$ problems\footnote{A formal definition of \textsc{MAX}-$\mathcal{SNP}$ problems is given in~\cite{papad}.} like \textsc{Max-CUT}, \textsc{Max}-$k$-\textsc{SAT} have been presented, by exploiting the structure of various classes of instances. One particularly significant line of research
is the study of the approximability of dense instances\footnote{For example, a dense instance of \textsc{Max}-$k$-\textsc{SAT} is an instance where the number of clauses is $\Omega(n^k)$.} of those problems, which was initiated by Arora, Karger and Karpinski~\cite{dense} and de la Vega~\cite{vega}. This line of work has produced several results in approximating dense instances of $\mathcal{NP}$-hard problems~\cite{dense_2, dense_1, dense_3, dense_4}. More specifically, in~\cite{dense} a framework was presented which shows that a family of problems, including \textsc{Max-CUT} and \textsc{Max}-$k$-\textsc{SAT}, admits a PTAS on dense instances. 
They actually gave additive approximations for the problems, which can be made into a multiplicative $1-\epsilon$ approximation due to the denseness of each problem. 
The framework was later extended and generalized to solve almost-sparse instances of the same problems by using subexponential time~\cite{fotakis}. 

\subsection{Our contribution} 
The first goal of this paper is to utilize the additional power given by a small (logarithmic) number of binary predictions to design a learning-augmented algorithm that significantly improves the running time of the PTAS of~\cite{dense} for dense instances of the following problems\footnote{Note that computing an optimal solution for all these problems remains $\mathcal{NP}$-hard even for dense instances~\cite{dense}. Moreover, they have no FPTAS unless $\mathcal{P}= \mathcal{NP}$ (deterministic), or $\mathcal{NP} \subseteq \mathcal{BPP}$ (randomized). 
}:
\begin{itemize}
    \item \textbf{\textsc{Max-CUT}:} Given an undirected graph $G=(V,E)$, partition the vertices of the graph into two complementary sets so as to maximize the number of edges with exactly one vertex in each set.

    \item \textbf{\textsc{Max-DICUT}:}
    The directed version of \textsc{Max-CUT}. Given a directed graph $G=(V,E)$, find a subset $T \subseteq V$ of vertices to maximize the total number of edges $(u,v)$ 
    with $u \in T$ and $v \in \overline{T}$.

    \item \textbf{\textsc{Max-HYPERCUT}($d$):} A natural generalization of \textsc{Max-CUT} to hypergraphs of dimension $d$. In \textsc{Max-HYPERCUT} an edge is considered cut if it has at least one endpoint on each side.

    \item \textbf{$k$-\textsc{Densest Subgraph}:}
    Given an undirected graph $G$, find a subset $C$ of $k$ vertices so that the induced subgraph $G[C]$ has a maximum number of edges.

    \item \textbf{\textsc{Max}-$k$-\textsc{SAT}:}
    Given an instance with $n$ variables that consists of $m$ boolean clauses $f_1,\dots, f_m$, each clause being a disjunction of at most $k$ literals, 
    we seek a truth assignment to the variables that maximizes the number of satisfied clauses.
    
\end{itemize}
We consider in this work a new learning augmented framework (see Section~\ref{notation} for precise definitions), called Learning Augmented Approximation (LAA), where we want to get approximation ratios close to $1-\epsilon$, smoothly decreasing with the error made in the predictions, while having a small (polynomial here, with no dependency on $\epsilon$) time complexity. We are particularly focusing on parsimonious predictions, using typically a logarithmic number of prediction bits.

Let us start with $\textsc{Max-CUT}$ for a smoother exposition (Section~\ref{maxcut}). For dense \textsc{Max-CUT}, the PTAS of Arora, Karger and Karpinski~\cite{dense} gives for any user-defined $\epsilon >0$, an $1-\epsilon$ randomized approximation that runs in time $n^{O(1/ \epsilon^2)}$. Since the work of~\cite{dense}, there has been faster PTAS, like the one in~\cite{claire}, but all with an exponential dependence on $1/ \epsilon$ in the running time. In this work, we use predictions to improve the running time for $\textsc{Max-CUT}$ to a low-degree polynomial with  no dependency on  $\epsilon$ while getting an approximation ratio $1-\epsilon-f( \textit{error})$, for a linear function $f$ with respect to the prediction \textit{error} (Theorem~\ref{ptas_0}).

More precisely, given $\epsilon > 0$, we sample a set $S$ of $O(\log n / \epsilon^3)$ vertices and get a binary prediction $\hat{a_i} \in \{0, 1\}$ for the placement of each vertex $i$ (side of the cut) of the sample at an optimal solution $\textbf{a}=(a_1,\dots,a_n)$. The prediction \textit{error} is just the sum of the absolute differences $|\hat{a_i}-a_i|$ of the variables in the sample $S$. Dealing with the LAA framework, we design an algorithm \textsc{LAA-Cut} which approximates \textsc{MAx-CUT} as follows, where $T_{LP}$ denotes the time to solve an LP with $n$ variables and $O(n)$ constraints. 
\begin{restatable}{theorem}{MaxcutTime}
 Let $G$ a $\delta$-dense graph. Then, for any $\epsilon > 0$ with  $|S|=  \Theta(\ln n / (\epsilon^3 \delta^4))$, \textsc{LAA-Cut} runs in time $O(n \cdot T_{LP})$ and, with respect to the approximation ratio, is with high probability $(1-\epsilon)$-consistent, $\big(1-\epsilon-8\frac{\textit{error}}{\delta |S|}\big)$-smooth and $0.878$-robust, where \emph{error} is the prediction error.    
\end{restatable}

Note that the density condition is necessary to achieve a consistency of $1-\epsilon$. Indeed, otherwise, using exhaustive search we would get a PTAS for all MaxCut instances 
while this problem is APX-hard.

In Section~\ref{general}, we generalize the approach applied to \textsc{MaxCUT}. Retracing the steps in~\cite{dense} we express each problem as a maximization problem of a low degree polynomial with bounded coefficients and $n$ binary variables. Then, we recursively
decompose the polynomial problem into lower-degree polynomials estimating the coefficients by using predictions on a sample of variables. In the end, we get an integer linear program, for which we obtain a fractional solution in polynomial time. Using randomized rounding we obtain an integer solution for the original problem. The running time of our algorithm with predictions is much shorter than the running time of the PTAS (Theorem~\ref{general_theorem}). The algorithm we get can be seen as an additive approximation (depending on the prediction \textit{error}), but translates into a multiplicative one when applied to dense instances of the problems studied (Section~\ref{applications}). As for MaxCut, the density condition is necessary for the result to hold.

Using the same approach as for \textsc{Max-CUT}, we obtain an algorithm corresponding to the LAA framework (Theorem~\ref{general-fixed-time}), which can be applied to all our problems (and possibly many more). Here again, we emphasize the fact that we use the predictions parsimoniously, which is highly desirable as a predictor is typically a machine learned model that can be computationally expensive.

While we acknowledge that there might not be a readily available oracle setting for the specific problems under study, we propose considering a scenario similar to the pricing policy implemented by OpenAI for ChatGPT. OpenAI charges customers based on the number of tokens (words) in both the input and output\footnote{\url{https://openai.com/pricing}}. Similarly, one could consider a pricing policy for machine learning models in a private company that tackle computational problems. Employing predictions parsimoniously in that case would result in cost savings, as it would require $O(\log n)$ instead of $O(n)$ tokens.

Let us also note that our work can be easily extended to the multiple predictions setting~\cite{multiple}. Instead of receiving only one prediction for the values of $S$, we receive $k$ different predictions from different predictors. Running our algorithm $k$ times (with the same $S$) and outputting the best solution, we get an approximation with respect to the best predictor, i.e., the one with the lowest prediction \textit{error}. The time overhead is just a multiplicative $k$.

\section{Notation and Preliminaries} \label{notation}
We start by giving a definition of 
density for each problem studied in this work. 

\begin{definition}
    An undirected graph $G(V,E)$ with $n$ vertices is $\delta$-dense when $\delta = \frac{2|E|}{n(n-1)}$. For a directed graph $G(V,E)$, $\delta = \frac{|E|}{n(n-1)}$.  A dimension-$d$ hypergraph is $\delta$-dense if it has at least $\delta n^d$ edges. Similarly, a $k$-\textsc{SAT} formula is $\delta$-dense when it has at least $\delta n^k$ clauses.
\end{definition}

In the paper, we assume that $\delta, d, k$ are constants. As explained in the introduction, we deal with optimizing polynomials. The following definition will be particularly useful in the next sections.
\begin{definition}
    A Polynomial Integer Program (PIP) is of the form
    \begin{align*}
        &\text{max } p(x_1,\dots,x_n)\\
        s.t. \qquad &l_i \le p_i(\textbf{x}) \le u_i  &&i=1,\dots,m \\
         &x_i \in \{0,1\} &&\forall i \le n,
    \end{align*}
    where $p, p_1, \dots, p_n$ are polynomials. The PIP could have minimization instead of maximization. If all $p, p_i$ have degree at most $d$, we call this program a degree-$d$ PIP.
\end{definition}

Let us now define a class of PIPs that are easier to approximate. Note that solving PIPs is $\mathcal{NP}$-hard in general.

\begin{definition}
    A degree-$d$ polynomial is $c$-smooth (or it has smoothness $c$) if the absolute value of each coefficient of each degree $i$ monomial is at most $c\cdot n^{d-i}$.
    
    A $c$-smooth degree-$d$ PIP is a PIP in which the objective function and the constraints are $c$-smooth polynomials with degree at most $d$.
\end{definition}
We assume that $c, d$ are constants.

\paragraph*{Prediction Model}
Let $\mathcal{I}$ be an instance of an optimization problem and $S$ a subset of the variables of the problem sampled uniformly at random. Then, we are given predictions on the  values of the  variables in $S$ at an optimal solution for the instance $\mathcal{I}$ (i.e., a prediction value for each distinct variable of $S$). In order to measure the quality of the predictions, we define the prediction \textit{error}.
\begin{definition} {(Prediction \textit{error})}
    Let $S \subseteq \{1, 2, \dots, n\}$ be a multiset and fix an optimal solution $\textbf{a}=(a_1, \dots, a_n) \in \{0,1\}^n$ for an optimization problem. Given a prediction $\hat{a_j} \in \{0,1\}$ for every distinct $a_j$, $\forall j \in S$ \big(at most $|S|$ predictions in total\big) we define the prediction \textit{error} as follows
    \begin{equation*} 
        \textit{error} = \sum_{j \in S} |\hat{a_j} - a_j| = \sum_{j \in S} \textit{error}_j,
    \end{equation*}
    where $\textit{error}_j = |\hat{a_j} - a_j|, \forall j \in S$.
\end{definition}
The \textit{error} is the absolute error, and $\frac{\textit{error}}{|S|}$ is the relative prediction error. In our algorithms, for ease of explanation we sample $S$ uniformly at random with replacement. Throughout the paper we omit the fact that our predictions may use only a subset of $S$ (distinct elements of $S$). 

\paragraph*{Learning-augmented approximation framework.}

Learning-augmented algorithms  have three main properties that we adjust in the context of this work.
In the Learning Augmented Approximation (LAA) framework, we say that a (randomized) algorithm is:
\begin{itemize}
    \item  $\alpha$-\textit{consistent}, if it is an $\alpha$-approximation with high probability when $ \textit{error} = 0$, 
    \item $\beta$-\textit{robust}, if it is a $\beta$-\textit{approximation} with high probability regardless of the value of \textit{error}, and 
    \item $\gamma$-\textit{smooth} for a continuous function $\gamma(\textit{error})$, if it is a $\gamma(\textit{error})$-approximation with high probability.
\end{itemize}
Note that the smoothness of a polynomial has no connection with the smoothness of a learning-augmented algorithm. In the paper, the distinction between the two notions will be made clear due to the context of each sentence.

\paragraph*{Notation} In the following, we use $a\pm b$ as a shorthand for the interval $[a-b,a+b]$, for $a,b \ge 0$.  Moreover, with $[l,u]\pm a$, where $l<u$ and $a\ge 0$, we denote the interval $[l-a,u+a]$. Finally, we often use $|OPT|$ to denote the value of the optimal solution of an optimization problem.

\section{\textsc{Max-CUT}} \label{maxcut}
In this section, we introduce our approach and apply it to \textsc{Max-CUT} in a graph $G(V,E)$ that is $\delta$-dense. First, we show how to speed up the PTAS of~\cite{dense} using a limited number of predictions (Section~\ref{Max-cut_1}-~\ref{maxcut_proof}), leading to algorithm \textsc{LA-PTAS-Cut}. Then, we use this algorithm to construct the algorithm \textsc{LAA-Cut} (Section~\ref{maxcut_time}).

\subsection{Overview of Algorithm \textsc{LA-PTAS-Cut}} \label{Max-cut_1}
First, let us write \textsc{Max-CUT} as follows:
\begin{align}
& \text{max } p(\textbf{x}) = \sum_{i=1}^{n} x_i \cdot \sum_{j \in N(i)}(1-x_j) \nonumber \\
& \text{s.t. } x_i \in \{0,1\} \, \forall{i}\nonumber
\end{align}
where $N(i)$ denotes the set of neighbors of vertex $i$.
The vector $\textbf{x} \in \{0.1\}^n$ characterizes a cut: $x_i=1$ (resp. $x_i=0$) indicates that vertex $i$ is placed on the right (resp. left) side of the cut. The objective function $p(\textbf{x})$ is an $n$-variate degree-2 2-smooth polynomial.

The above formulation is a quadratic integer program and cannot be approximated efficiently. Our goal is to turn it into a linear program. For that reason, we set 
$r_i(\textbf{x}) = \sum_{j \in N(i)}x_j$
and rewrite the \textsc{Max-CUT} problem in the following way:
\begin{align} \label{maxcut_qp}
    &\text{max } p(\textbf{x}) = \sum_{i=1}^{n} x_i \cdot \big(|N(i)| -r_i(\textbf{x})\big) \\
    &\text{s.t. } x_i \in \{0,1\} \, \forall{i}. \nonumber
\end{align}
We first define the algorithm \textsc{LA-PTAS-Cut} which will approximate these $r_i(\textbf{x})$'s, that are linear functions, by using sampling on the vertices of $G$ and a prediction on the values of the sample at the optimal solution. The main idea~\cite{dense} is that if we have a good estimation of the value of each $\rho_i = r_i(\textbf{a})$ at the optimal solution $\textbf{a}$, then we can approximately solve~\eqref{maxcut_qp}. Let us write the Integer Linear Program using our estimates $\hat{e_i}$ for $r_i(\textbf{x})$ (which will be obtained using the predictions):
\begin{gather*} 
    \text{max } p(\textbf{x}) = \sum_{i=1}^{n} x_i \cdot \big( |N(i)|-\hat{e_i}\big)\qquad \text{(IP)}\\
    \text{s.t. } \sum_{j \in N(i)}x_j \ge \hat{e_i}-f(\textit{error}, \epsilon, \delta) \cdot n \quad\forall i\in V \\
    \sum_{j \in N(i)}x_j \le \hat{e_i}+f(\textit{error}, \epsilon, \delta) \cdot n \quad\forall i\in V\\
    x_i \in \{0,1\} \quad  \forall{i}. 
\end{gather*}
The estimated values $\hat{e_i}$ and the values $f(\textit{error}, \epsilon, \delta)$ are computed such that the optimal solution $\textbf{a}$ is a feasible solution to the above (IP). Note that we can replace the right-hand-side of the first $n$ constraints by $\text{max}\{0, \hat{e_i}-f(\textit{error}, \epsilon, \delta) \cdot n\}$ (as $\sum_{j \in N(i)}x_j\ge 0$) and the one of next $n$ constraints by $\text{min}\{|N(i)|, \hat{e_i}+f(\textit{error}, \epsilon, \delta) \cdot n\}$ (as $\sum_{j \in N(i)}x_j \le |N(i)|$). Let (LP) denote the Linear Programming relaxation of (IP), i.e., setting each $x_i \in [0,1]$.

Our learning-augmented algorithm \textsc{LA-PTAS-Cut} approximates the optimal solution of \textsc{Max-CUT} executing the following steps:
\begin{itemize}
    \item We sample a set of vertices $S$ and get a prediction on the value of each $x_i, \forall i \in S$ at the optimal solution $\textbf{a}=(a_1,\dots,a_n)$. Using these values we then estimate the values of $r_i(\textbf{a}), \forall i \in \{1, \dots, n\}$ at the optimal solution $\textbf{a}$ (Section~\ref{maxcut_sec_1}).

    \item For each possible value of the integer variable \textit{error}, we perform the following two steps and output the best solution.

    \begin{enumerate}
    \item We replace each function $r_i$ by the corresponding estimate $\hat{e_i}$ of $r_i(\textbf{a})$, formulate (IP) and show that an optimal solution for this (IP) is a good approximation for \textsc{Max-CUT} (Section~\ref{maxcut_sec_2}). 

    \item Then, we find an optimal fractional solution $\textbf{y}$ to (LP) and obtain an integral solution $\textbf{z}$ by applying (naive) randomized rounding to $\textbf{y}$ (Section~\ref{maxcut_sec_3}).
    \end{enumerate}
\end{itemize}

\subsection{Estimating Coefficients via Sampling and Predictions} \label{maxcut_sec_1}
We take a random sample $S \subseteq V$ of $O(\log n)$ vertices. Assume for now that we know the values $a_j$ at the optimal cut for all sampled vertices $j$. Using the Sampling lemma for \textsc{Max-CUT} from~\cite{fotakis}\footnote{We use the sampling lemma of Fotakis et. al. for a more straightforward analysis.} with these values $a_j$ we can compute an estimate $e_i = \sum_{j \in S  \cap N(i)} a_j \cdot n / |S|$ of each $\rho_i=r_i(\textbf{a}) = \sum_{j \in N(i)}a_j$ for every vertex $i$ such that $e_i \approx \rho_i$ with high probability. Let us now state the Sampling lemma for \textsc{Max-CUT} of~\cite{fotakis} and show how to get the estimates $e_i$ for $\rho_i$ if we know the values $a_j$'s at the optimal solution.

\begin{lemma}{Sampling lemma~\cite{fotakis}} \label{sampling_lemma_1}
    Let $\textbf{a}$ be a binary vector and $G(V,E)$ be a $\delta$-dense graph. For $\epsilon > 0$, we let $g=\Theta(1/\epsilon^3)$ and $S$ be a multiset of $|S| = g \ln n / \delta$ vertices chosen uniformly at random with replacement from $V$. For any vertex $i$, if $e_i = (n/|S|) \sum_{j \in N(i) \cap S} a_j$ and $\rho_i = \sum_{j \in N(i)} a_j$, with probability at least $1-2/n^3$,
    \begin{equation*}
        (1-\epsilon)e_i - \epsilon \delta n \le \rho_i \le (1+\epsilon)e_i + \epsilon \delta n.
    \end{equation*}
\end{lemma} 
Using Lemma~\ref{sampling_lemma_1} we get that with probability at least $1-2/n^3$,
\begin{gather}
     e_i -\epsilon \cdot e_i - \epsilon \delta n \le \rho_i \le e_i+\epsilon \cdot e_i + \epsilon \delta n \nonumber\\
     \implies e_i -(\epsilon + \epsilon \delta) n \le \rho_i \le e_i+(\epsilon + \epsilon \delta) n \nonumber \\
     \implies e_i -2\epsilon n \le \rho_i \le e_i+2\epsilon n, \label{eq_2}
\end{gather}
since $\delta \le 1$ and assuming wlog that $|e_i| \le n$ (since $|\rho_i| \le n$).
Taking the union bound over all vertices, we have that \eqref{eq_2} holds for all vertices $i \in V$ simultaneously with probability at least $1-2/n^2$.

Of course, the problem is that we do not know the values $a_j, \forall j \in S$. In~\cite{dense} they try all possible ($2^{O(\log n)} = n^{O(1)}$, as $|S|=g\ln n /\delta =O(\ln n)$ for fixed $\epsilon,\delta$) placements of the vertices in the sample, so they guess all $a_j$ correctly. Here, we get a prediction $\hat{a_j}$ for each $a_j$. Using these predicted values we compute an estimate $\hat{e_i} = \sum_{j \in S  \cap N(i)} \hat{a_j} \cdot n / |S|$ for each $\rho_i$. It is easy to see that
\begin{align} 
     &\sum_{j \in S \cap N(i)} \hat{a_j} \in \sum_{j \in S \cap N(i)} a_j \pm \textit{error} \nonumber \\
     &\implies  \frac{n \sum_{j \in S \cap N(i)} \hat{a_j}}{|S|} \in \frac{n \sum_{j \in S \cap N(i)} a_j}{|S|} \nonumber \pm \frac{n}{|S|} \textit{error} \nonumber \\ 
     &\implies \hat{e_i} \in e_i \pm \frac{n}{|S|} \textit{error} \nonumber\\
     &\implies e_i - \frac{n}{|S|} \textit{error} \le \hat{e_i} \le e_i + \frac{n}{|S|} \textit{error} \nonumber \\
     &\implies \hat{e_i} - \frac{n}{|S|} \textit{error} \le e_i \le \hat{e_i} + \frac{n}{|S|} \textit{error}. \label{maxcut_3}
\end{align}
Using~\eqref{eq_2} and~\eqref{maxcut_3} we get that for all vertices $i \in V$ with probability at least $1-2/n^2$,
\begin{equation} \label{eq_4}
    \hat{e_i} -\bigg(2\epsilon + \frac{\textit{error}}{|S|}\bigg) n \le \rho_i \le \hat{e_i} +\bigg(2\epsilon + \frac{\textit{error}}{|S|}\bigg) n .
\end{equation}


\subsection{Formulating the Integer Linear Program} \label{maxcut_sec_2}
Now we can use the estimates $\hat{e_i}$ for the coefficients $\rho_i$ of the quadratic integer program of \textsc{Max-CUT} and write the following Integer Linear Program (IP):
\begin{gather*} 
    \text{max} \sum_i x_i \cdot \big(|N(i)|-\hat{e_i}\big) \qquad \text{(IP)} \\
    \text{s.t.} \sum_{j \in N(i)}x_j \ge \hat{e_i} -\bigg(2\epsilon  + \frac{\textit{error}}{|S|}\bigg) n \quad\forall i\in V \\
    \sum_{j \in N(i)}x_j \le \hat{e_i}+
    \bigg(2\epsilon  + \frac{\textit{error}}{|S|}\bigg) n \quad\forall i\in V \\
    x_i \in \{0,1\} \quad  \forall{i \in V}. 
\end{gather*}
Note again that we can replace the right-hand side of the first $n$ constraints by $\text{max}\{0, \hat{e_i}-
    \big(2\epsilon + \frac{\textit{error}}{|S|}\big) n \}$ and the one of the next $n$ constraints by $\text{min}\{|N(i)|, \hat{e_i}+
    \big(2\epsilon + \frac{\textit{error}}{|S|}\big) n \}$. With probability at least $1-2/n^2$, the previous Integer Linear Program (IP) is feasible, since the optimal solution $\textbf{a}$ satisfies it. The only problem is that we do not know the value of the $\textit{error}$. To overcome this issue, we try all possible values and guess it. Note that $\textit{error} \le |S|$, so the runtime overhead is at most $|S| \le n$ (actually, even $|S|\leq g\ln n/\delta$). From now on we assume that we know the true value of the $\textit{error}$.

    Let $\textbf{z}$ be an optimal solution to this (IP). We show that $\textbf{z}$ is a good approximation for the optimal solution $\textbf{a}$ of \textsc{Max-CUT}. We have that
    \begin{gather} 
        \sum_{i \in V} z_i \sum_{j \in N(i)} (1-z_j)
            = \sum_{i \in V} z_i \big(|N(i)|-\sum_{j \in N(i)} z_j \big) \nonumber\\
            \ge \sum_{i \in V} z_i \bigg(|N(i)|-\big(\hat{e_i}+
    \big(2\epsilon + \frac{\textit{error}}{|S|}\big) n \big) \bigg), \nonumber \\
    \text{by the constraints of (IP),}\nonumber \\
    \ge \sum_{i \in V} z_i (|N(i)|-\hat{e_i}) - \bigg(2\epsilon  + \frac{\textit{error}}{|S|}\bigg) n^2 \nonumber   \\
    \ge \sum_{i \in V} a_i (|N(i)|-\hat{e_i}) - \bigg(2\epsilon  + \frac{\textit{error}}{|S|}\bigg) n^2, \nonumber  \\
    \text{since \textbf{z} is an integer optimal solution of (IP),} \nonumber \\
     \ge \sum_{i \in V} a_i \bigg(|N(i)|-\rho_i -\bigg(2\epsilon  + \frac{\textit{error}}{|S|}\bigg) n\bigg) \nonumber \\ - \bigg(2 \epsilon + \frac{\textit{error}}{|S|}\bigg)  n^2 \nonumber
     \text{, from~\eqref{eq_4},} \nonumber \\
    \ge \sum_{i \in V} a_i \big( |N(i)|- \rho_i \big) - 2\bigg(2\epsilon  + \frac{\textit{error}}{|S|} \bigg) n^2 \nonumber  \\
    = \textbf{a} - 2\bigg(2\epsilon + \frac{\textit{error}}{|S|} \bigg) n^2 \nonumber  \\
    = |OPT| - 2\bigg(2\epsilon  + \frac{\textit{error}}{|S|} \bigg) n^2 . \label{maxcut_opt}
    \end{gather}
Thus, with probability at least $1-2/n^2$, the integer optimal solution of (IP) is close to the optimum of \textsc{Max-CUT}.

\subsection{Randomized Rounding} \label{maxcut_sec_3}
Now we relax the integrality constraints, allowing $ x_i \in [0,1]$ and get the Linear Programming relaxation of (IP). We can solve (LP) via linear programming and obtain a fractional optimal solution $\textbf{y} \in [0,1]^n$. Then, we use randomized rounding to convert the fractional solution to an integral one with approximately the same cut value. To achieve that we will use the following lemma, which is due to Raghavan and Thomson~\cite{raghavan} and Arora et al.~\cite{dense}.

\begin{lemma} {Randomized Rounding~\cite{raghavan, dense}}
If $c$ and $f$ are positive integers and $0 < \epsilon <1$, then the following is true for any integers $n > 0$. Let $\textbf{y} = (y_i)$ be a vector of $n$ variables, $0 \le y_i \le 1$, that satisfies a certain linear constraint $\textbf{a}^T\textbf{y}=b$, where each $|a_i| \le c$. Construct a vector $\textbf{z} = (z_i)$ randomly by setting $z_i=1$ with probability $y_i$ and 0 with probability $1-y_i$. Then, with probability at least $1-n^{-f}$, we have that
$$ \textbf{a}^T \textbf{z} \in b \pm c \sqrt{fn \ln n} .$$    
\end{lemma}

Since each $r_i(\textbf{x})$ is a linear function with $0/1$ coefficients, it follows from the Randomized Rounding lemma and the union bound that with probability at least $1-n^{-f+1}$ holds (for every vertex simultaneously) that
\begin{equation} \label{maxcut_4}
    r_i(\textbf{z}) \in r_i(\textbf{y}) \pm O(\sqrt{n} \ln n) \qquad \forall i\in V.
\end{equation}
Additionally, since each $|N(i)|-r_i(\textbf{y})$ is at most $n$, we can use again the Randomized Rounding lemma to get that with probability at least $1-n^{-f}$
\begin{equation} \label{maxcut_5}
    \sum_{i \in V} z_i \big( |N(i)|-r_i(\textbf{y}) \big) \in \sum_{i \in V} y_i \big( |N(i)|-r_i(\textbf{y})\big) \pm O(n^{3/2} \ln n). 
\end{equation}
Both inequalities hold simultaneously with probability at least $1-n^{-f+1}-n^{-f} \approx 1-n^{-f+1}$. 

So, when~\eqref{maxcut_opt}, ~\eqref{maxcut_4} and~\eqref{maxcut_5} hold we get:
\begin{gather*}
    \sum_{i \in V} z_i \big( |N(i)|-\sum_{j \in N(i)} z_j \big) =  \sum_{i \in V} z_i \big( |N(i)|- r_i(\textbf{z}) \big) \\
    \ge \sum_{i \in V} z_i \big(|N(i)|-r_i(\textbf{y})-O(\sqrt{n} \ln n)\big), \text{ from~\eqref{maxcut_4},} \\
    \ge \sum_{i \in V} z_i \big(|N(i)|-r_i(\textbf{y})\big)-O(n^{3/2} \ln n)\\
    \ge \sum_{i \in V} y_i \big( |N(i)|-r_i(\textbf{y})\big) - O(n^{3/2} \ln n) \text{ from~\eqref{maxcut_5},}\\
    \ge |OPT| -  2\bigg(2\epsilon  + \frac{\textit{error}}{|S|} \bigg) n^2 - o(1) n^2.
\end{gather*}
The last inequality is due to~\eqref{maxcut_opt} and the fact that the fractional optimal solution $\textbf{y}$ cannot be worse than the integer optimal solution $\textbf{z}$.

Finally, all estimations are good with probability at least $1-2/n^2$ and the randomized rounding works with probability at least $\approx 1-1/n^{-f+1}$. Thus, our learning augmented approximation scheme works with probability $\approx 1-2/n^2-1/n^{f-1}$. 

\subsection{Analysis of \textsc{LA-PTAS-Cut}} \label{maxcut_proof}

To conclude, we state and prove the formal theorem for \textsc{LA-PTAS-Cut}.

\begin{theorem} \label{ptas_0}
     Let $G$ a $\delta$-dense graph. Then, for any $\epsilon > 0$ with $|S|=  \Theta(\ln n / (\epsilon^3 \delta^4))$, \textsc{LA-PTAS-Cut} runs in time $O(n \cdot T_{LP})$ and is an $\big(1-\epsilon-8\frac{\textit{error}}{\delta |S|}\big)$-approximation for \textsc{Max-CUT} with probability at least $1-3/n^2$, where \emph{error} is the prediction error. 
\end{theorem}

\begin{proof}
    For any $\epsilon > 0$, using \textsc{LA-PTAS-Cut} with $f=3$, $\epsilon' = \epsilon \delta/16, g=1/\epsilon'^3$ and sample size $|S|=g \ln n/\delta= \Theta \big(\ln n / \epsilon^3 \delta^4 \big)$, we get a cut $\textbf{z}$ that with probability at least $\approx 1-3/n^2$ satisfies:
\begin{gather*}
    p(\textbf{z}) \ge |OPT|- 2\bigg( 2\epsilon'+\frac{\textit{error}}{|S|}\bigg) n^2 \\
    = |OPT|- \bigg(16 \epsilon'/\delta +\frac{8\textit{error}}{\delta |S|}\bigg) \delta n^2 /2 \\
    = |OPT|- \bigg(\epsilon +\frac{8\textit{error}}{\delta |S|}\bigg) \delta n^2 /4 \\
    \ge \bigg(1 - \epsilon -  8\frac{\textit{error}}{\delta |S|}\bigg) \cdot |OPT|, \\
    \text{as $|OPT|$ is at least $|E|/2=\delta n(n-1)/2 \ge \delta n^2/4$, } \forall n \ge 2.
\end{gather*}
Therefore, the approximation ratio of the algorithm is $\big(1 - \epsilon -  8\frac{\textit{error}}{\delta |S|}\big)$ with high probability. 

Regarding its running time, it is easy to see that it only requires to solve a linear program with $n$ variables and $O(n)$ constraints for each possible value of the \textit{error} (i.e., $|S|\leq n$ rounds). So, it runs in time $O(n \cdot T_{LP})$, where $T_{LP}$ is the time to solve an LP with $n$ variables and $O(n)$ constraints. There are many algorithms that solve linear programs~\cite{Vaidya, Cohen, lee}. The state of the art does it in time $O^*(n^w)$, where $w$ is the matrix multiplication exponent (the current value is $w \approx 2.38$)~\cite{start}.
\end{proof}

\subsection{LAA Framework} \label{maxcut_time}
We now describe the algorithm~\textsc{LAA-Cut} that combines \textsc{LA-PTAS-Cut} (for consistency and smoothness) and a known polytime constant-approximation algorithm (for robustness). 

\paragraph*{Consistency \& Smoothness.}
For any $\epsilon > 0$, using \textsc{LA-PTAS-Cut} with predictions on a sample of size $|S|= \Theta \big(\ln n / \epsilon^3 \delta^4 \big)$ we have an algorithm with approximation ratio of $\big(1 - \epsilon -  8\frac{\textit{error}}{\delta |S|}\big)$ (Theorem~\ref{ptas_0}). The algorithm is randomized and gives with probability at least $1-3/n^2$ the aforementioned approximation ratio that depends on $\epsilon$, which is user-defined, the density of the graph $\delta$ and the $error$ of our prediction (consistency and smoothness of the approximation ratio). The algorithm runs in $O(n \cdot T_{LP})$. Note that the original algorithm in~\cite{dense}\footnote{Even the best PTAS for \textsc{Max-CUT} runs in time that depends exponentially on $1/\epsilon$.} runs in time dominated by the exhaustive search which takes time $O(2^{1/(\epsilon \delta)^2 \log n}) = n^{1/ (\epsilon \delta)^2}$, which depends on $\epsilon, \delta$. For example, for $\epsilon = 1-0.878 = 0.122$ (to obtain the approximation ratio of the algorithm by Goemans and Williamson~\cite{goemans}) the running time is $O(n^{67})$!

Additionally, we would like to mention that if the (partial) prediction corresponds to a (global) prediction with approximation ratio $\alpha$, then our algorithm has approximation ratio at least $\alpha-\epsilon$ with high probability. This is a direct consequence of our proof. The relative (partial) error is an unbiased estimator of the relative (global) error.

\paragraph*{Robustness.}
In case the \textit{error} of our predictions is too large, we would like to be able to ensure an approximation guarantee for the value of the cut (robustness of the approximation ratio). We can do that  by running in parallel the celebrated algorithm of Goemans and Williamson~\cite{goemans} which achieves an approximation ratio of $\approx 0.878$. The algorithm can run in time $\tilde{O}(n^2)$ using the Arora-Kale algorithm~\cite{kale, trevisan}. Of course, the algorithm is randomized and we should demand that both algorithms (\textsc{LA-PTAS-Cut} and that of Goemans and Williamson) succeed simultaneously. Note that the algorithm can also be derandomized~\cite{derandomized}.

Therefore, the approximation ratio of our learning-augmented scheme is $\max \{1 - \epsilon -  8\frac{\textit{error}}{\delta |S|}, 0.878 \}$ (with probability at least $1-3/n^2$) for a $\delta$-dense graph and runs in time $O(n \cdot T_{LP})$. Note that for different values of the parameter $\epsilon>0$, the prediction \textit{error} is not the same due to the change of the sampling size.
Consequently, we restate the theorem for \textsc{LAA-Cut}.

\MaxcutTime*


\begin{proof}
    The proof becomes now trivial. Here, we just clarify the success probability of the algorithm. The algorithm of Goemans and Williamson succeeds with probability at least $\tau$, for some constant $\tau>0$. \textsc{LA-PTAS-Cut} as well, so we just have to run both algorithms a constant number of times independently to get a success with  probability at least $1-\eta$, for arbitrarily small $\eta$. We can also boost the success probability to $1-1/\Omega(n)$ by running the algorithm a logarithmic number of times independently.  
\end{proof}

\section{Smooth Polynomial Integer Programs} \label{general}
In this section, we extend the approach applied to \textsc{Max-CUT} to approximately optimize $c$-smooth polynomials of degree $d$ over all binary vectors $\textbf{x} \in \{ 0, 1\}^n$, as done in~\cite{dense}. We exploit the fact that smooth polynomial integer programs can be recursively decomposed into lower degree PIPs to eventually obtain a linear program. We can assume wlog that instead of solving the optimization problem, we can deal with the feasibility problem of a smooth polynomial integer program, i.e., given a feasible PIP find an integer solution that is approximately feasible. Our general learning-augmented algorithm \textsc{LA-PTAS} follows the same steps as the version for \textsc{Max-CUT} generalizing our approach\footnote{In this section, we focus on maximization problems. The minimization version can be handled similarly.}. Next, we will use \textsc{LA-PTAS} to get our algorithm for PIPs. 

We only briefly sketch here the general approach to build \textsc{LA-PTAS}. Details of the construction and proofs are deferred to the appendix. As shown in~\cite{dense}, each absolute value of a $c$-smooth polynomial of degree $d$ is bounded by $2cen^d$ (where $\ln e =1$). Thus, the optimal value of a PIP is not too large and we can reduce the optimization of a PIP $p(\textbf{x})$ to the feasibility version of the problem using binary search. Specifically, it is sufficient to find if the problem $p(\textbf{x}) \ge M$ for $M>0$ has a feasible solution. The parameter $M>0$ can be computed by using binary search in $(0, 2cen^d]$ taking at most $O(\log (2cen^d))=O(\log (c n^d))$ runs of the algorithm. Throughout the section we denote by $N$ the set $N=\{1,\dots,n\}$.

Furthermore, another key idea to generalize our results is the following decomposition lemma of a degree-$d$ $c$-smooth polynomial.

\begin{lemma} {\cite{dense}} \label{form}
A $c$-smooth polynomial $p$ of degree $d$ on $\textbf{x} = (x_1, \dots, x_n)$ can be written uniquely as
\begin{equation*}
    p(\textbf{x}) = t+ \sum_{i} x_i p_i(x_i, \dots, x_n) 
\end{equation*}
where $t$ is a constant and each $p_i$ is a $c$-smooth polynomial of degree $d-1$ and depends only on variables with index $i$ or greater.
\end{lemma}

Assume that we would like to optimize $p(\textbf{x})$ of degree $d$. We transform the optimization problem into a feasibility one using binary search as discussed previously. We now have the feasibility problem $p(\textbf{x}) \ge M$ for a known $M>0$. Using the decomposition lemma we can write $p(\textbf{x})$ as $p(\textbf{x}) = t+ \sum x_i p_i(\textbf{x})$. Computing an estimate $\hat{e_i}$ of the value of $p_i(\textbf{a})$ at the optimal solution $\textbf{a}$, we replace the degree $d$ constraint $p(\textbf{x}) \ge M$ with $t+ \sum x_i \hat{e_i} \ge M$ and a family of constraints on the values $p_i(\textbf{x})$. Then, we recursively expand these degree $d-1$ constraints, continuing until all constraints are linear. We can compute the estimations $\hat{e_i}$ of $p_i(\textbf{a})$ by writing $p_i(\textbf{x}) = \sum x_j p_{ij}(\textbf{x})$. We then recursively estimate the values $p_{ij}(\textbf{a})$, and use sampling and the predicted values $\hat{a_k}, k\in S$ to estimate $p$ based on the values of $p_{ij}$. Thus we end up with an Integer Linear Program (see Appendix~\ref{appsubsec:ilp}). 

Then, we relax the integral constraints and solve the Linear Programming relaxation of ($d$-IP). Finally, we use randomized rounding to get an integral solution. The details and proofs of each step can be found in Appendices~\ref{appsubsec:B2} (Estimating Polynomials via Sampling and Predictions),~\ref{appsubsec:B3} (Transforming degree $d$ constraints into linear constraints) and~\ref{appsubsec:B4} (Randomized Rounding for Smooth Polynomials). We finally obtain the following result (see Appendix~\ref{pip_2} for the proof). In the following, $T'_{LP}$ denotes the time to solve an LP with $n$ variables and $poly(n)$ constraints.

\begin{theorem} \label{general_theorem}
    Given a feasible $c$-smooth degree-$d$ PIP with $n$ variables, its objective function $p$ and $m=poly(n)$ constraints, \textsc{LA-PTAS} finds a binary solution $\textbf{z}$ with probability at least $1/2$ such that 
    \begin{equation*}
        p(z_1,\dots,z_n) \ge |OPT| - \bigg(\epsilon+ 4ced(d-1)\frac{\textit{error}}{|S|}\bigg) n^{d},
    \end{equation*}
    given predictions on the values of the distinct variables of $S$ at the optimal solution (optimum of PIP) that is chosen uniformly at random (with replacement), where $|S| = \Theta(\frac{ c^4 f d^7}{\epsilon^3} \ln n)$, where $f>0$ is such that $n^f = \Theta(m \cdot n^d)$. The running time of the algorithm is $O\big(n \ln (cn^d) \cdot T'_{LP} \big)$.     
\end{theorem}

\paragraph*{LAA Framework}  \label{pip_3}
Let us now describe our algorithm \textsc{LAA-General} that relies on \textsc{LA-PTAS}, whose ratio depends on the \textit{error} while guaranteeing a fixed (polynomial) running time.

{\it Consistency \& Smoothness.} For any $\epsilon >0$, we use \textsc{LA-PTAS} with $|S| = \Theta(\frac{ c^4 f d^7}{\epsilon^3} \ln n)$ and get an algorithm that with high probability outputs a value of at least $|OPT| - \big(\epsilon+ 4ced(d-1)\frac{\textit{error}}{|S|}\big) n^{d}$ (Theorem~\ref{general_theorem}). As we will see in Section~\ref{applications}, this additive approximation leads to a multiplicative one when the instance is dense for the problems we consider. Therefore, we achieve the desired consistency and smoothness in the approximation ratio. The running time of \textsc{LA-PTAS} is $O\big(n \ln (cn^d)\cdot T'_{LP} \big)$, with no dependency on $\epsilon$ (compared to  exponential dependency in $1/\epsilon$ in the PTAS's). Furthermore, note again that our algorithm is guaranteed to have approximation ratio at least as good as the approximation quality of the (global) prediction (minus $\epsilon$).
 
{\it Robustness.}
When the \textit{error} of our predictions is too large, we can ensure an approximation guarantee for the solution value (robustness of the approximation ratio). We can achieve that by running in parallel a constant  approximation  algorithm for the given problem (if it exists in the literature). The running time, here, depends on the approximation algorithm that is used for the problem in question. If the algorithm is randomized, we take the joint probability that both algorithms succeed simultaneously.

\begin{theorem} \label{general-fixed-time}
We are given a feasible $c$-smooth degree-$d$ PIP with $n$ variables, its objective function $p$ and $m=poly(n)$ constraints. Let also $ALG$ be an algorithm  for the PIP that runs in time $T_{ALG}$ and produces a solution with cost at least $\alpha |OPT|$. For any $\epsilon >0$ with $|S| = \Theta\big(\frac{ c^4 f d^7}{\epsilon^3} \ln n\big)$, where $f>0$ is such that $n^f = \Theta(m  n^d)$, \textsc{LAA-General} runs in time $\max \{ O\big(n \ln (cn^d)\cdot T'_{LP} \big), T_{ALG}\big\}$ and outputs with high probability a solution with cost at least
$$\max\bigg\{|OPT| - \bigg(\epsilon+ 4ced(d-1)\frac{\textit{error}}{|S|}\bigg) n^{d}, \alpha |OPT| \bigg\},$$
where \emph{error} is the prediction error.
\end{theorem}

\section{Applications} \label{applications}
In this section, we explain how to apply the algorithm of Section~\ref{general}. We give the example of \textsc{Max}-$k$-\textsc{SAT}, while problems \textsc{Max-DICUT}, \textsc{Max-HYPERCUT}($d$) and $k$-\textsc{Densest Subgraph} are deferred to Appendix~\ref{appsec:app}. Note also that the algorithm can be applied to the more general \textsc{Max}-$k$-\textsc{CSP}, as shown in~\cite{dense}.

\paragraph*{\textsc{Max}-$k$-\textsc{SAT}}

A standard arithmetization technique (see~\cite{dense}) can be used to reduce any instance
of \textsc{Max}-$k$-\textsc{SAT} with $n$ variables to solving a degree-$k$ polynomial $p(\textbf{x})$ such that the optimal truth assignment for \textsc{Max}-$k$-\textsc{SAT} corresponds to an $\textbf{a} \in \{0,1\}^n$ that maximizes $p(\textbf{x})$. Moreover, the value of the
optimal \textsc{Max}-$k$-\textsc{SAT} solution is equal to $p(\textbf{a})$.

Let us now assume that the number of clauses is $m \ge \delta n^k$. The number of clauses of size exactly $k$ is $m-O(n^{k-1})$ and a random assignment satisfies each one of them with probability $1-2^{-k}$. Thus it follows that the maximum number of clauses that can be made true is
\begin{equation*}
    |OPT| \ge (1-2^{-k})(m-O(n^{k-1})).
\end{equation*}
Therefore, we use \textsc{LA-PTAS} with $\epsilon' = O(\epsilon /2^k)$ and get the desired accuracy for \textsc{Max}-$k$-\textsc{SAT}. 
For robustness, we can use the poly-time randomized $0.797$ approximation~\cite{avidor} to robustify, as done by \textsc{LAA-General}.

\section*{Acknowledgement}
This work was partially funded by the grant ANR-19-CE48-0016 from the French National Research Agency (ANR).

\bibliography{bibliography}

\begin{thebibliography}{10}

\bibitem{multiple}
Keerti Anand, Rong Ge, Amit Kumar, and Debmalya Panigrahi.
\newblock Online algorithms with multiple predictions.
\newblock In Kamalika Chaudhuri, Stefanie Jegelka, Le~Song, Csaba Szepesvari, Gang Niu, and Sivan Sabato, editors, {\em Proceedings of the 39th International Conference on Machine Learning}, volume 162 of {\em Proceedings of Machine Learning Research}, pages 582--598. PMLR, 17--23 Jul 2022.

\bibitem{hypergraph}
Gunnar Andersson and Lars Engebretsen.
\newblock Better approximation algorithms for set splitting and not-all-equal sat.
\newblock {\em Information Processing Letters}, 65(6):305--311, 1998.

\bibitem{antoniadis}
Antonios Antoniadis, Joan Boyar, Marek Elias, Lene~Monrad Favrholdt, Ruben Hoeksma, Kim~S. Larsen, Adam Polak, and Bertrand Simon.
\newblock Paging with succinct predictions.
\newblock In Andreas Krause, Emma Brunskill, Kyunghyun Cho, Barbara Engelhardt, Sivan Sabato, and Jonathan Scarlett, editors, {\em Proceedings of the 40th International Conference on Machine Learning}, volume 202 of {\em Proceedings of Machine Learning Research}, pages 952--968. PMLR, 2023.

\bibitem{antoniadis_2}
Antonios Antoniadis, Christian Coester, Marek Elias, Adam Polak, and Bertrand Simon.
\newblock Online metric algorithms with untrusted predictions.
\newblock In Hal~Daumé III and Aarti Singh, editors, {\em Proceedings of the 37th International Conference on Machine Learning}, volume 119 of {\em Proceedings of Machine Learning Research}, pages 345--355. PMLR, 13--18 Jul 2020.

\bibitem{kale}
Sanjeev Arora and Satyen Kale.
\newblock A combinatorial, primal-dual approach to semidefinite programs.
\newblock In {\em Proceedings of the Thirty-Ninth Annual ACM Symposium on Theory of Computing}, STOC '07, page 227–236, New York, NY, USA, 2007. Association for Computing Machinery.

\bibitem{dense}
Sanjeev Arora, David Karger, and Marek Karpinski.
\newblock Polynomial time approximation schemes for dense instances of {NP}-hard problems.
\newblock {\em Journal of Computer and System Sciences}, 58(1):193--210, 1999.

\bibitem{sudan}
Sanjeev Arora, Carsten Lund, Rajeev Motwani, Madhu Sudan, and Mario Szegedy.
\newblock Proof verification and the hardness of approximation problems.
\newblock {\em J. ACM}, 45(3):501–555, may 1998.

\bibitem{asahiro}
Yuichi Asahiro, Kazuo Iwama, Hisao Tamaki, and Takeshi Tokuyama.
\newblock Greedily finding a dense subgraph.
\newblock In Rolf Karlsson and Andrzej Lingas, editors, {\em Algorithm Theory - SWAT 1996, Proceedings}, Lecture Notes in Computer Science (including subseries Lecture Notes in Artificial Intelligence and Lecture Notes in Bioinformatics), pages 136--148. Springer Verlag, 1996.

\bibitem{avidor}
Adi Avidor, Ido Berkovitch, and Uri Zwick.
\newblock Improved approximation algorithms for max nae-sat and max sat.
\newblock In Thomas Erlebach and Giuseppe Persinao, editors, {\em Approximation and Online Algorithms}, pages 27--40, Berlin, Heidelberg, 2006. Springer Berlin Heidelberg.

\bibitem{bai}
Xingjian Bai and Christian Coester.
\newblock Sorting with predictions.
\newblock In {\em Thirty-seventh Conference on Neural Information Processing Systems}, 2023.

\bibitem{bampis}
Evripidis Bampis, Bruno Escoffier, Themis Gouleakis, Niklas Hahn, Kostas Lakis, Golnoosh Shahkarami, and Michalis Xefteris.
\newblock {Learning-Augmented Online TSP on Rings, Trees, Flowers and (Almost) Everywhere Else}.
\newblock In Inge~Li G{\o}rtz, Martin Farach-Colton, Simon~J. Puglisi, and Grzegorz Herman, editors, {\em 31st Annual European Symposium on Algorithms (ESA 2023)}, volume 274 of {\em LIPIcs}, pages 12:1--12:17, Dagstuhl, Germany, 2023. Schloss Dagstuhl -- Leibniz-Zentrum f{\"u}r Informatik.

\bibitem{evripidis_2}
Evripidis Bampis, Bruno Escoffier, and Michalis Xefteris.
\newblock Canadian traveller problem with predictions.
\newblock In Parinya Chalermsook and Bundit Laekhanukit, editors, {\em Approximation and Online Algorithms}, pages 116--133, Cham, 2022. Springer International Publishing.

\bibitem{dense_1}
Cristina Bazgan, Wenceslas Fernandez~de La~Vega, and Marek Karpinski.
\newblock {Polynomial time approximation schemes for dense instances of minimum constraint satisfaction}.
\newblock {\em {Random Structures and Algorithms}}, 23(1), 2003.

\bibitem{benomar}
Ziyad Benomar and Vianney Perchet.
\newblock Advice querying under budget constraint for online algorithms.
\newblock In {\em Thirty-seventh Conference on Neural Information Processing Systems}, 2023.

\bibitem{dense_4}
Jean Cardinal, Marek Karpinski, Richard Schmied, and Claus Viehmann.
\newblock Approximating vertex cover in dense hypergraphs.
\newblock {\em Journal of Discrete Algorithms}, 13:67--77, 2012.
\newblock Best Papers from the 3rd International Conference on Similarity Search and Applications (SISAP 2010).

\bibitem{Cohen}
Michael~B. Cohen, Yin~Tat Lee, and Zhao Song.
\newblock Solving linear programs in the current matrix multiplication time.
\newblock In {\em Proceedings of the 51st Annual ACM SIGACT Symposium on Theory of Computing}, STOC 2019, page 938–942, New York, NY, USA, 2019. Association for Computing Machinery.

\bibitem{cohenaddad}
Vincent Cohen-Addad, Tommaso d'Orsi, Anupam Gupta, Euiwoong Lee, and Debmalya Panigrahi.
\newblock Max-cut with $\epsilon$-accurate predictions, 2024.
\newblock \href {https://arxiv.org/abs/2402.18263} {\path{arXiv:2402.18263}}.

\bibitem{warmstart}
Michael Dinitz, Sungjin Im, Thomas Lavastida, Benjamin Moseley, and Sergei Vassilvitskii.
\newblock Faster matchings via learned duals.
\newblock In A.~Beygelzimer, Y.~Dauphin, P.~Liang, and J.~Wortman Vaughan, editors, {\em Advances in Neural Information Processing Systems}, 2021.

\bibitem{drygala}
Marina Drygala, Sai~Ganesh Nagarajan, and Ola Svensson.
\newblock Online algorithms with costly predictions.
\newblock In Francisco Ruiz, Jennifer Dy, and Jan-Willem van~de Meent, editors, {\em Proceedings of The 26th International Conference on Artificial Intelligence and Statistics}, volume 206 of {\em Proceedings of Machine Learning Research}, pages 8078--8101. PMLR, 25--27 Apr 2023.

\bibitem{clustering}
Jon~C. Ergun, Zhili Feng, Sandeep Silwal, David Woodruff, and Samson Zhou.
\newblock Learning-augmented k-means clustering.
\newblock In {\em International Conference on Learning Representations}, 2022.

\bibitem{maxdicut}
U.~Feige and M.~Goemans.
\newblock Approximating the value of two power proof systems, with applications to max 2sat and max dicut.
\newblock In {\em Proceedings Third Israel Symposium on the Theory of Computing and Systems}, pages 182--189, 1995.

\bibitem{subgraph}
Uriel Feige and Michael Langberg.
\newblock Approximation algorithms for maximization problems arising in graph partitioning.
\newblock {\em Journal of Algorithms}, 41(2):174--211, 2001.

\bibitem{vega}
W.~Fernandez de~la Vega.
\newblock Max-cut has a randomized approximation scheme in dense graphs.
\newblock {\em Random Structures \& Algorithms}, 8(3):187--198, 1996.

\bibitem{dense_2}
W.~Fernandez de~la Vega and M.~Karpinski.
\newblock Polynomial time approximation of dense weighted instances of max-cut.
\newblock {\em Random Structures \& Algorithms}, 16(4):314--332, 2000.

\bibitem{fotakis}
Dimitris Fotakis, Michael Lampis, and Vangelis~Th. Paschos.
\newblock {Sub-exponential Approximation Schemes for CSPs: From Dense to Almost Sparse}.
\newblock In Nicolas Ollinger and Heribert Vollmer, editors, {\em 33rd Symposium on Theoretical Aspects of Computer Science (STACS 2016)}, volume~47 of {\em Leibniz International Proceedings in Informatics (LIPIcs)}, pages 37:1--37:14, Dagstuhl, Germany, 2016. Schloss Dagstuhl -- Leibniz-Zentrum f{\"u}r Informatik.

\bibitem{constraint}
Suprovat Ghoshal, Konstantin Makarychev, and Yury Makarychev.
\newblock Constraint satisfaction problems with advice, 2024.
\newblock \href {https://arxiv.org/abs/2403.02212} {\path{arXiv:2403.02212}}.

\bibitem{goemans}
Michel~X. Goemans and David~P. Williamson.
\newblock Improved approximation algorithms for maximum cut and satisfiability problems using semidefinite programming.
\newblock {\em J. ACM}, 42(6):1115–1145, nov 1995.

\bibitem{kumar}
Sungjin Im, Ravi Kumar, Aditya Petety, and Manish Purohit.
\newblock Parsimonious learning-augmented caching.
\newblock In Kamalika Chaudhuri, Stefanie Jegelka, Le~Song, Csaba Szepesvari, Gang Niu, and Sivan Sabato, editors, {\em Proceedings of the 39th International Conference on Machine Learning}, volume 162 of {\em Proceedings of Machine Learning Research}, pages 9588--9601. PMLR, 17--23 Jul 2022.

\bibitem{dense_3}
Tomokazu Imamura and Kazuo Iwama.
\newblock Approximating vertex cover on dense graphs.
\newblock In {\em Proceedings of the Sixteenth Annual ACM-SIAM Symposium on Discrete Algorithms}, SODA '05, page 582–589, USA, 2005. Society for Industrial and Applied Mathematics.

\bibitem{facility}
Shaofeng H.-C. Jiang, Erzhi Liu, You Lyu, Zhihao~Gavin Tang, and Yubo Zhang.
\newblock Online facility location with predictions.
\newblock In {\em International Conference on Learning Representations}, 2022.

\bibitem{start}
Shunhua Jiang, Zhao Song, Omri Weinstein, and Hengjie Zhang.
\newblock Faster dynamic matrix inverse for faster lps, 2020.
\newblock \href {https://arxiv.org/abs/2004.07470} {\path{arXiv:2004.07470}}.

\bibitem{khot}
Subhash Khot.
\newblock Ruling out ptas for graph min‐bisection, dense k‐subgraph, and bipartite clique.
\newblock {\em SIAM Journal on Computing}, 36(4):1025--1071, 2006.

\bibitem{svensson}
Silvio Lattanzi, Ola Svensson, and Sergei Vassilvitskii.
\newblock Speeding up {B}ellman ford via minimum violation permutations.
\newblock In Andreas Krause, Emma Brunskill, Kyunghyun Cho, Barbara Engelhardt, Sivan Sabato, and Jonathan Scarlett, editors, {\em Proceedings of the 40th International Conference on Machine Learning}, volume 202 of {\em Proceedings of Machine Learning Research}, pages 18584--18598. PMLR, 23--29 Jul 2023.

\bibitem{lee}
Yin~Tat Lee and Aaron Sidford.
\newblock Efficient inverse maintenance and faster algorithms for linear programming.
\newblock In {\em 2015 IEEE 56th Annual Symposium on Foundations of Computer Science}, pages 230--249, 2015.

\bibitem{lu}
Pinyan Lu, Xuandi Ren, Enze Sun, and Yubo Zhang.
\newblock Generalized sorting with predictions.
\newblock In {\em Symposium on Simplicity in Algorithms (SOSA)}, pages 111--117, 2021.

\bibitem{lykouris}
Thodoris Lykouris and Sergei Vassilvitskii.
\newblock Competitive caching with machine learned advice.
\newblock {\em J. ACM}, 68(4), jul 2021.

\bibitem{derandomized}
Sanjeev Mahajan and Ramesh Hariharan.
\newblock Derandomizing semidefinite programming based approximation algorithms.
\newblock {\em Proceedings of IEEE 36th Annual Foundations of Computer Science}, pages 162--169, 1995.

\bibitem{claire}
Claire Mathieu and Warren Schudy.
\newblock Yet another algorithm for dense max cut: go greedy.
\newblock In {\em Proceedings of the Nineteenth Annual ACM-SIAM Symposium on Discrete Algorithms}, SODA '08, page 176–182, USA, 2008. Society for Industrial and Applied Mathematics.

\bibitem{mitzenmacher}
Michael Mitzenmacher.
\newblock {Scheduling with Predictions and the Price of Misprediction}.
\newblock In Thomas Vidick, editor, {\em 11th Innovations in Theoretical Computer Science Conference (ITCS 2020)}, volume 151 of {\em Leibniz International Proceedings in Informatics (LIPIcs)}, pages 14:1--14:18, Dagstuhl, Germany, 2020. Schloss Dagstuhl -- Leibniz-Zentrum f{\"u}r Informatik.

\bibitem{papad}
Christos Papadimitriou and Mihalis Yannakakis.
\newblock Optimization, approximation, and complexity classes.
\newblock In {\em Proceedings of the Twentieth Annual ACM Symposium on Theory of Computing}, STOC '88, page 229–234, New York, NY, USA, 1988. Association for Computing Machinery.

\bibitem{kumar_1}
Manish Purohit, Zoya Svitkina, and Ravi Kumar.
\newblock Improving online algorithms via ml predictions.
\newblock In S.~Bengio, H.~Wallach, H.~Larochelle, K.~Grauman, N.~Cesa-Bianchi, and R.~Garnett, editors, {\em Advances in Neural Information Processing Systems}, volume~31. Curran Associates, Inc., 2018.

\bibitem{raghavan}
Prabhakar Raghavan and Clark~D. Thomson.
\newblock Randomized rounding: A technique for provably good algorithms and algorithmic proofs.
\newblock {\em Combinatorica}, 7:365--374, 1985.

\bibitem{trevisan}
Luca Trevisan.
\newblock Max cut and the smallest eigenvalue.
\newblock {\em SIAM Journal on Computing}, 41(6):1769--1786, 2012.

\bibitem{Vaidya}
P.M. Vaidya.
\newblock Speeding-up linear programming using fast matrix multiplication.
\newblock In {\em 30th Annual Symposium on Foundations of Computer Science}, pages 332--337, 1989.

\end{thebibliography}

\newpage
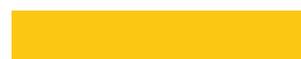
\appendix

\section{Missing material of Section~\ref{general}}

\subsection{Integer Linear Program}\label{appsubsec:ilp}

\begin{gather*}
    \sum_{j \in N} x_j \hat{e_j} \ge M  \qquad \text{($d$-IP)} \\
    t_{i_1} + \sum_{j \in N} x_j \hat{e}_{i_1 j} \in \hat{e}_{i_1} \pm f( \textit{error}, \epsilon, \delta) n^{d-1} \qquad \forall i_1 \in N \\
    t_{i_1 i_2} + \sum_{j \in N} x_j \hat{e}_{i_1 i_2 j} \in \hat{e}_{i_1 i_2} \pm f(\textit{error}, \epsilon, \delta) n^{d-2} \\ \forall (i_1, i_2) \in N \times N \\
    \dots \\
    t_{i_1 \dots i_{d-\ell}} + \sum_{j \in N} x_j \hat{e}_{i_1 \dots i_{d-\ell} j} \in \hat{e}_{i_1 \dots i_{d-\ell}} \pm f(\textit{error}, \epsilon, \delta) n^{d-\ell} \\
    \forall (i_1,\dots, i_{d-\ell})\in N^{d-\ell} \\
    \dots \\
    t_{i_1 \dots i_{d-1}} + \sum_{j \in N} x_j \hat{e}_{i_1 \dots i_{d-1} j} \in \hat{e}_{i_1 \dots i_{d-1}} \pm f(\textit{error}, \epsilon, \delta) n \\
    \forall (i_1,\dots, i_{d-1})\in N^{d-1} \\
    x_j \in \{0,1\} \qquad \forall j \in N.
\end{gather*}

\subsection{Estimating Polynomials via Sampling and Predictions}\label{appsubsec:B2}
We show how to estimate the coefficients $p_i(\textbf{a})$ at an optimal solution that are smooth polynomials of degree at most $d-1$ using sampling and predictions. This step is required to be able to replace the constraint on $p(\textbf{x})$ by linear constraints. We describe an algorithm \textsc{Evaluate}, adaptation of the one in~\cite{dense}, which can approximate the value of a $c$-smooth degree-$d$ polynomial $p(x_1,\dots, x_n)$ on any unknown binary vector $\textbf{a} = (a_1, \dots, a_n)$ given a partial prediction about $\textbf{a}$. So, we take a sample $S$ from $\{1, \dots, n\}$, uniformly at random and with replacement. The sample size is $O(\log n)$. Next, we get a  predicted value $\hat{a_j}$ for each value $a_j, \forall j \in S$ at the optimal solution $\textbf{a}$ (note that the number of predicted values is at most $|S|$). Showing that \textsc{Estimate} (Algorithm~\ref{algo:evaluate}) provides us with good estimates for the coefficients $p_i(\textbf{a})$ becomes easier using the following General Sampling lemma of~\cite{fotakis}.

\begin{lemma}{General Sampling lemma~\cite{fotakis}} \label{general_sampling}
    Let $\textbf{a} \ in \{0,1\}^n$ and let $(\rho_j)_{j \in N}$ be any sequence such that for some integer $d \ge 0$ and some constant $\beta \ge 1$, $|\rho_j| \le (d+1) \beta n^d$, $\forall j \in N$. For all integers $f \ge 1$ and $\epsilon > 0$, let $g = \Theta(fd\beta/\epsilon^3)$ and $S$ a multiset of $|S| = g \ln n$ indices chosen uniformly at random with replacement from $N$. If $(n/|S|) \sum_{j \in S} \rho_j a_j$, $\rho = \sum_{j \in N} \rho_j a_j$ and $\bar{\rho} = \sum_{j \in N} |\rho_j|$, with probability at least $1-4/n^{f+1}$,
    \begin{equation*}
        \rho - \epsilon \bar{\rho} - \epsilon n^{d+1} \le (n/|S|) \sum_{j \in S} \rho_j a_j \le \rho + \epsilon \bar{\rho}+\epsilon n^{d+1}.
    \end{equation*}
\end{lemma}

\begin{algorithm}
    \caption{\textsc{Evaluate}($p, S, \{ \hat{a_i}: i\in S \}$)}
    \label{algo:evaluate}
\begin{algorithmic}
    \REQUIRE polynomial $p$ of degree at most $d$, \\
    set of variables indices $S$,\\
    predictions $\hat{a_i}$ for $i \in S$.    
    \ENSURE Estimate for $p(a_1, \dots, a_n)$.
    \IF{$\text{deg($p$)} = 0$} 
        \STATE \textbf{return} $p$
    \ELSE
        \STATE $p(x_1, \dots, x_n) = t + \sum x_i p_i(x_1, \dots, x_n)$
        \FOR{$i \in S$}
            \STATE $\hat{e_i} \gets \textsc{Evaluate}(p_i, S,\{ \hat{a_i}: i\in S \})$
        \ENDFOR
        \STATE \textbf{return } $t+ (n/|S|) \sum_{i\in S} \hat{a_i} \hat{e_i}$
    \ENDIF
\end{algorithmic}
\end{algorithm}

Let us now show the following lemma about \textsc{Evaluate} with set $S$ and the predictions.

\begin{lemma} \label{general_estimation}
    Let $p$ be a $c$-smooth degree-$d$ polynomial in $n$ variables $x_i$ and $\textbf{a}=(a_1, \dots, a_n) \in \{0,1\}^n$. Let $f\ge 1$ be an integer, $\epsilon >0$, $\beta \ge \max\{1,2ce\}$, and $S$ be a set of $O(g \ln n)$ indices chosen randomly and with replacement with $g=\Theta(fd\beta/\epsilon^3)$. Also, let $\hat{a_j} ,\forall j \in S$ be a prediction on the values of $a_j, \forall j \in S$. Then, with probability at least $1-4/n^{f+1-d}$, set $S$ is such that \textsc{Evaluate}$(p,S, \{ \hat{a_i}: i\in S\})$ returns a value in  
    \begin{equation}
        p(a_1,\dots,a_n) \pm \bigg((2ce+1) d \epsilon+ 2ced\frac{\textit{error}}{|S|}\bigg) n^d.
    \end{equation}
\end{lemma}

\begin{proof}
    The proof is by induction on the degree $d$. For the case $d=0$ we have by definition that $\textit{error}=0$ and \textsc{Evaluate} returns a value that is exactly $p(\textbf{a})=t$.

    For the inductive step let $\rho_i = p_i(a_1, \dots, a_n)$. So,
    \begin{equation*}
        p(\textbf{a}) = t + \sum_{i=1}^n a_i \rho_i.
    \end{equation*}
    Note that each $p_i$ has degree at most $d-1$.
    
    First, we apply the General Sampling lemma (Lemma~\ref{general_sampling}) for $p(\textbf{a})$ with $d' = d-1$, $\beta \ge \max\{1,2ce\}$, $g = \Theta(f d \beta / \epsilon^3)$ and $|S| = g \ln n$. As each $p_i$ is a $c$-smooth degree-$(d-1)$ polynomial we have that $|\rho_i| \le 2ce n^{d-1}\leq (d'+1)\beta n^{d'}$. Then, we have with probability at least $1-4/n^{f+1}$ that
    \begin{equation*}
        (n/|S|) \sum_{i \in S} a_i \rho_i \in \sum_{i \in N} a_i \rho_i \pm (\epsilon \sum_{i \in N} |\rho_i| + \epsilon n^{d}).
    \end{equation*}

    Using again $|\rho_i| \le 2ce n^{d-1}$  we get:
    \begin{equation} \label{prob_1}
        (n/|S|) \sum_{i \in S} a_i \rho_i \in \sum_{i \in N} a_i \rho_i \pm \epsilon (2ce + 1) n^d.
    \end{equation}

    Given the predictions $(\hat{a_j} \in \{0, 1\}, \forall j \in S$ we have that $\textit{error}_j = |\hat{a_j} - a_j|, \forall j \in S$. Note that $\textit{error} = \sum_{i \in S} \textit{error}_i$. Using $\hat{a_j} \in a_j \pm \textit{error}_j, \forall j \in S$, we have that 
    \begin{gather}
        \frac{n}{|S|} \sum_{i \in S} \hat{a_i} \rho_i \in \frac{n}{|S|} \sum_{i \in S} a_i \rho_i \pm \frac{n}{|S|} \sum_{i \in S} \textit{error}_i \rho_i \nonumber\\
        \subseteq \frac{n}{|S|} \sum_{i \in S} a_i \rho_i \pm \frac{n}{|S|} 2cen^{d-1} \sum_{i\in S} \textit{error}_i \nonumber \\
        = \frac{n}{|S|} \sum_{i \in S} a_i \rho_i \pm \frac{2ce}{|S|} \textit{error} \cdot n^{d}. \label{eq_pred}
    \end{gather}

    Let $\hat{e_i} = \textsc{Evaluate}(p_i, S, \{\hat{a_i: i \in S})$.
     By the inductive hypothesis, \textsc{Evaluate} outputs estimates $\hat{e_i}$ such that
     \begin{gather}         
         \rho_i \in \hat{e_i} \pm \left((2ce+1)(d-1) \epsilon n^{d-1} + 2ce(d-1)\frac{\textit{error}}{|S|} n^{d-1}\right) \nonumber\\ 
         \text{or equivalently} \nonumber\\
         \hat{e_i} \in \rho_i \pm \left((2ce+1)(d-1) \epsilon n^{d-1} + 2ce(d-1)\frac{\textit{error}}{|S|} n^{d-1}\right) \label{eq_est}
     \end{gather}
     with probability at least $1-4/n^{f+1-(d-1)}=1-4/n^{f+2-d}$. Taking the union bound all $n$, values $\rho_i$ are (simultaneously) estimated to within this bound with probability at least $$1-n\cdot 4/n^{f+2-d} =1-4/n^{f+1-d}.$$ So, together with \eqref{prob_1} we get with probability at least $1-4/n^{f+1-d}-4/n^{f+1} \approx 1 - 4/n^{f+1-d}$ the following:
     \begin{gather*}
         t+\frac{n}{|S|}\sum_{i \in S} \hat{a_i} \hat{e_i} \\
         \in t
         +\frac{n}{|S|} \sum_{i \in S} \hat{a_i} \big(\rho_i \pm (2ce+1)(d-1) \epsilon n^{d-1} + 2ce(d-1)\frac{\textit{error}}{|S|} n^{d-1}\big)
         \text{, by~\eqref{eq_est},}\\
        \subseteq t
         +\frac{n}{|S|} \sum_{i \in S} \hat{a_i} \rho_i \pm  (2ce+1)(d-1) \epsilon n^d + 2ce(d-1)\frac{\textit{error}}{|S|} n^d 
         \text{, due to $\sum_{i \in S} \hat{a_i} \le |S| \le n$,}\\
         \subseteq t
         +\frac{n}{|S|} \sum_{i \in S} a_i \rho_i \pm 2ce\frac{\textit{error}}{|S|} n^d  \pm  (2ce+1)(d-1) \epsilon n^d + 2ce(d-1)\frac{\textit{error}}{|S|} n^d
         \text{, by \eqref{eq_pred}.} 
     \end{gather*}

     So, we get that 
     \begin{gather*}
         t+\frac{n}{|S|}\sum_{i \in S} \hat{a_i} \hat{e_i} \\
         \subseteq t
         +\frac{n}{|S|} \sum_{i \in S} a_i \rho_i \pm  (2ce+1)(d-1) \epsilon n^d + 2ced\frac{\textit{error}}{|S|} n^d \\
         \subseteq t
         + \sum_{i \in N} a_i \rho_i \pm (2ce+1)\epsilon n^d \pm  (2ce+1)(d-1) \epsilon n^d + 2ced\frac{\textit{error}}{|S|} n^d 
         \text{, by~\eqref{prob_1},}\\
         \subseteq t
         + \sum_{i \in N} a_i \rho_i \pm  (2ce+1) d \epsilon \cdot n^d + 2ced\frac{\textit{error}}{|S|} \cdot n^d.
     \end{gather*}
\end{proof}

\subsection{Transforming degree $d$ constraints into linear constraints}\label{appsubsec:B3}

First, let us observe that the previous proof for \textsc{Evaluate} shows implicitly that \textsc{Evaluate} estimates the values of all polynomials arising from the decomposition of a polynomial $p$ with the described accuracy with probability at least $1-4/n^{f+1-d}$. Specifically, it estimates every polynomial of degree $d'$ to within $\big((2ce+1) d' \epsilon+ 2ced'\frac{\textit{error}}{|S|}\big) n^{d'}$.

\begin{algorithm}
    \caption{\textsc{Linearize}\big($L \le p(\textbf{x}) \le U, S$,
    $\{ \hat{a_i}: i\in S \}, \epsilon$\big)}
    \label{algo:linearize}
\begin{algorithmic}
    \REQUIRE constraint involving polynomial $p$ of degree $d$, \\
    set of variables indices $S$,\\
    predictions $\hat{a_i}$ for $i \in S$\\
    parameter $\epsilon>0$.    
    \ENSURE A set of linear constraints.
    \IF{$p$ is linear} 
        \STATE \textbf{output} the input constraint $L \le p(\textbf{x}) \le U$
    \ELSE 
        \STATE Out $\leftarrow$ $\emptyset$
        \STATE $p(x_1, \dots, x_n) = t + \sum x_i p_i(x_i, \dots, x_n)$
        \FOR{$i \in \{1, 2, \dots, n\}$}
            \STATE $\hat{e_i} \gets \textsc{Evaluate}(p_i, S,\{ \hat{a_i}: i\in S \})$\\
            $l_i \gets \hat{e_i} - \big((2ce+1) (d-1) \epsilon+ 2ce(d-1)\frac{\textit{error}}{|S|}\big) n^{d-1}$\\
            $u_i \gets \hat{e_i} + \big((2ce+1) (d-1) \epsilon+ 2ce(d-1)\frac{\textit{error}}{|S|}\big) n^{d-1}$\\
            Out$\leftarrow$ Out $\cup$ \textsc{Linearize}\big($l_i \le p_i(x_i,\dots,x_n) \le u_i$,
            $S, \{ \hat{a_i}: i\in S \}, \epsilon$\big)
        \ENDFOR
        \STATE \textbf{output } Out $\cup \{$ 
        \STATE  $t+\sum x_i \hat{e_i} \ge L - \big((2ce+1) d \epsilon+ 2ced\frac{\textit{error}}{|S|}\big) n^d $,      
        \STATE $t+\sum x_i \hat{e_i} \le U + \big((2ce+1) d \epsilon+ 2ced\frac{\textit{error}}{|S|}\big) n^d \}$
    \ENDIF
\end{algorithmic}
\end{algorithm}

We now use a modified version of algorithm \textsc{Linearize} (see Algorithm~\ref{algo:linearize}) given in~\cite{dense} to transform any polynomial constraint into a family of linear constraints. \textsc{Linearize} is a recursive algorithm that uses \textsc{Evaluate} to output linear constraints. It is easy to see that \textsc{Linearize} outputs a set of at most $2 n^{d-1}$ linear constraints such that the optimal solution $\textbf{a}$ satisfies all constraints and $\textbf{a}$ is a feasible solution to ($d$-IP), as long as \textsc{Evaluate} estimates all polynomials with the required accuracy. This happens since the decompositions of the polynomials are unique and common between the two algorithms. Thus, by the observation stated previously, we have that with probability at least $1-4d/ n^{f+1-d}$ the linear constraints output by \textsc{Evaluate} are jointly feasible. Let us now see why (using induction on the degree $d$). It is obviously true for $d=1$. Assume it is for $d-1\geq 1$. Then for a polynomial $p$ of degree $d$, \textsc{Linearize} outputs at most $n$ sets of constraints associated to polynomial of degree $d-1$, and two ``new'' constraints associated to $p$. By union bound (and recursive argument), they are simultaneously satisfied with probability at least $1-n4(d-1)/n^{f+1-(d-1)}-4/n^{f+1-d}=1-4d/n^{f+1-d}$.

Next, we show (as in~\cite{dense}) that any feasible solution to the linear system output by \textsc{Evaluate} is also an approximate solution to the input constraint (degree $d$ polynomial constraint). 

\begin{lemma} \label{error_dependence}
    Every feasible solution $(y_i) \in [0,1]^n$ to the set of linear constraints output by \textsc{Linearize} satisfies (without any assumption on the success of the sampling for set $S$):
    \begin{equation*}
        p(\textbf{y}) \in [L,U] \pm \bigg((4ce+2) d(d-1) \epsilon+ 4ced(d-1)\frac{\textit{error}}{|S|}\bigg) n^{d}.
    \end{equation*}
\end{lemma}

\begin{proof}
    We show the lemma by induction on degree $d$. The base case $d=1$ is trivial. For the inductive step we have that $d\ge 2$ and that $\textbf{y}$ is feasible for all constraints output by \textsc{Linearize}. By the inductive hypothesis we get for each $i$:
    \begin{gather*}
        p_i(\textbf{y}) \in [l_i,u_i] \pm \bigg((4ce+2) (d-1)(d-2) \epsilon+ 4ce(d-1)(d-2)\frac{\textit{error}}{|S|}\bigg) n^{d-1}.
    \end{gather*}
    So, it follows by substituting $l_i, u_i$ that

    \begin{gather}
        p_i(\textbf{y}) \in \hat{e_i} \pm \bigg((2ce+1) (d-1) \epsilon+ 2ce(d-1)\frac{\textit{error}}{|S|}\bigg) n^{d-1} \nonumber \\ \pm \bigg((4ce+2) (d-1)(d-2) \epsilon+ 4ce(d-1)(d-2)\frac{\textit{error}}{|S|}\bigg) n^{d-1} \nonumber\\
        \subseteq \hat{e_i}
        \pm \bigg((2ce+1)(d-1) (2d-3) \epsilon + 2ce(d-1)(2d-3)\frac{\textit{error}}{|S|}\bigg) n^{d-1}. \label{eq_p(y)} 
    \end{gather}

    Thus, 
    \begin{gather*}
        p(\textbf{y}) = t + \sum y_i p_i (y_i, \dots, y_n) \\
        \subseteq t + \sum y_i \cdot \bigg(\hat{e_i} 
        \pm \bigg((2ce+1) (d-1)(2d-3) \epsilon+ 2ce(d-1)(2d-3)\frac{\textit{error}}{|S|}\bigg) n^{d-1} \bigg)
        \text{, by~\eqref{eq_p(y)},}
        \\ 
        \subseteq t + \sum y_i \hat{e_i} 
        \pm \bigg((2ce+1) (d-1)(2d-3) \epsilon+ 2ce(d-1)(2d-3)\frac{\textit{error}}{|S|}\bigg) n^{d}    \\        
        \subseteq [L,U] \pm \bigg((2ce+1) d \epsilon+ 2ced\frac{\textit{error}}{|S|}\bigg) n^{d}
        \pm \bigg((2ce+1) (d-1)(2d-3) \epsilon\\+ 2ce(d-1)(2d-3)\frac{\textit{error}}{|S|}\bigg) n^{d},   \\
        \text{from the fact that $y$ is feasible for the constraint} \\ \text{which was output by \textsc{Linearize} before recursion,}\\
        \subseteq [L,U]  \pm \bigg( (2ce+1)\epsilon (2d^2-4d+3) + 2ce(2d^2-4d+3) \frac{\textit{error}}{|S|}\bigg) n^d \\
        \subseteq [L,U]  \pm \bigg( 2(2ce+1) d (d-1) \epsilon  + 4ced (d-1) \frac{\textit{error}}{|S|}\bigg) n^d.
    \end{gather*}
    The last inequality holds, since $d \ge 2 > 3/2$.    
\end{proof}

\subsection{Randomized Rounding for Smooth Polynomials}\label{appsubsec:B4}
Finally, we have to round our fractional solution to get an integral one. Using a lemma from~\cite{dense}, which shows that the randomized rounding outputs an integer value which is close to the fractional one for every $c$-smooth degree-$d$ polynomial, we conclude the last step of \textsc{LA-PTAS}. We restate the lemma for completeness.

\begin{lemma} {Randomized rounding for degree-$d$ polynomials~\cite{dense}} \label{general_round}
    Let $p$ be a $c$-smooth degree-$d$ polynomial. Let $\textbf{y} \in [0,1]^n$ be such that $p(y_1,\dots,y_n)=b$. Performing randomized rounding on $y_i$ to yield a $0,1$ vector $(z_i)$ we get that with probability at least $1-n^{d-f}$ we have that
    \begin{equation}
        p(z_1, \dots, z_n) \in [b \pm gd n^{d-1/2} \sqrt{\ln n}],
    \end{equation}
    where $g = 2ce\sqrt{f}$.
\end{lemma}

\subsection{Proof of Theorem~\ref{general_theorem}}  \label{pip_2}

\begin{proof}
    We have a feasible $c$-smooth degree-$d$ PIP with $m=\text{poly($n$)}$ constraints each one of which has degree at most $d$. Let $\textbf{a}= (a_1, \dots, a_n) \in \{0,1\}^n$ be a feasible solution. Here, we focus on the maximization problem (the minimization one is similar) of a polynomial $p(\textbf{x})$, where $x \in \{0,1 \}^n$. 

    Assume that we have found the optimal value $|OPT|>0$ of $p$ using binary search in time $O(\log cn^d)$. Then we write the maximization problem of $p$ as a feasibility one with $p(\textbf{x}) \ge |OPT|$, which is of course feasible. 
    
    Let $f>0$ be such that $n^f = 2m(n+4d)n^d/n = \Theta(m \cdot n^d)$.
    We let $\epsilon'=\frac{\epsilon}{(4ce+2)d(d-1)}$, $g=\Theta(cfd/\epsilon'^3)$ and $|S|=g \ln n$. Then, we take a random sample $S$ of variables with replacement and we are given a prediction $\hat{a_i}$  on the values $a_i$ for each $i \in S$. We use \textsc{Linearize} with error parameter $\epsilon'$ and replace each degree $d'$ constraint with $O(n^{d'-1})$ linear constraints. Therefore, we construct a linear integer system with $O(m \cdot n^{d-1})$ constraints. This new system is feasible with probability at least $1-4md/n^{f+1-d}$, since $\textbf{a}$ is an optimal solution to the PIP.

    Let us now relax the integrality constraint of each variable and solve the linear system with $n$ variables and $O(m \cdot n^{d-1})$ constraints in time $T'_{LP}$. From Lemma~\ref{error_dependence} for the fractional solution $\textbf{y}$ we get that the following holds:
    \begin{equation*}
        p(\textbf{y}) \ge |OPT|- \bigg((4ce+2) d(d-1) \epsilon'+ 4ced(d-1)\frac{\textit{error}}{|S|}\bigg) n^{d}.
    \end{equation*}

    Next, we use randomized rounding to get an integer solution $\textbf{z}$ that increases the additive loss by at most $O(n^{d-1/2} \sqrt{\ln n})=o(n^d)$. The rounding from Lemma~\ref{general_round} works simultaneously for all $m$ constraints with probability at least $1-m/n^{f-d}$.

    Consequently, our randomized learning-augmented approximation scheme works with probability at least 
    $1-m/n^{f-d}-4md/n^{f+1-d} > 1/2$ and outputs a solution such that 
        \begin{equation*}
        |\textsc{LA-PTAS}| \ge |OPT| - \bigg(\epsilon+ 4ced(d-1)\frac{\textit{error}}{|S|}\bigg) n^{d},
    \end{equation*}
    where $|S| = \Theta(\frac{128 c^4 e^4 f d^7}{\epsilon^3} \ln n) = \Theta(\frac{ c^4 f d^7}{\epsilon^3} \ln n)$.

    For the running time, we have to also guess the value of the \textit{error} which takes at most $n$. So, in total the general algorithm \textsc{LA-PTAS} runs in time  $O\big(n \ln (cn^d) \cdot T'_{LP} \big)$.
\end{proof}

\section{Missing material from Section~\ref{applications}}\label{appsec:app}

\paragraph*{\textsc{Max-DICUT}}
Let us write \textsc{Max-DICUT} of a directed graph $G=(V,E)$ as an $1$-smooth degree-$2$ polynomial integer program as follows:
\begin{align*}
    &\text{max }\sum_{(i, j) \in E} (1-x_i)x_j\\
    &\text{s.t. } x_i \in \{0,1\} \, \forall{i}.
\end{align*}
In a directed graph with density $\delta$, the value of the maximum cut is at least $\delta n^2 / 4$. Using \textsc{LA-PTAS} and Theorem~\ref{general_theorem}
with $\epsilon' = \delta \epsilon/4$ we find a cut of value at least
\begin{equation*}
    |OPT|\bigg(1 - \epsilon - 32e \frac{\textit{error}}{\delta |S|}\bigg),
\end{equation*}
where $|S| = O(\ln n /(\epsilon^3 \delta^3))$ and the running time is $O(\ln n \cdot T'_{LP})$. For the robustness of \textsc{LAA-General}, there is a polynomial time randomized approximation algorithm that achieves a ratio of $0.859$~\cite{maxdicut}. The rest of the steps to construct our two learning-augmented schemes for \textsc{Max-DICUT} follow trivially.

\paragraph*{\textsc{Max-HYPERCUT}($d$)}
We can formulate the problem as a smooth degree-$d$ PIP. Given an edge (set of vertices) $S'$, we use the term $1-\prod_{i \in S'} x_i - \prod_{i \in S'} (1-x_i)$, which is $1$ if $S'$ is cut and $0$ otherwise. Moreover, for the robustness of \textsc{LAA-General} we can use the randomized poly-time algorithm with approximation ratio of $0.72$~\cite{hypergraph}.

\paragraph*{$k$-\textsc{Densest Subgraph}}

Let $k \ge \gamma n$. If the graph is $\delta$-dense, using an averaging argument we have that the optimal solution contains at least $\gamma^2 \delta n^2/2$ edges. The problem is equivalent to maximizing the following degree-$2$ $1$-smooth PIP:
\begin{gather*}
    \text{max } p(\textbf{x})= \sum_{\{i, j\} \in E} x_i x_j\\
    \text{s.t. } \sum_{i=1}^n x_i =k\\
    \qquad x_i \in \{0,1\}.
\end{gather*}
We use the general algorithm \textsc{LA-PTAS} with $\epsilon' = \epsilon \gamma^2 \delta /2$ and noticing that our solution $z$ satisfies $\sum_{i=1}^n x_i \in [k \pm g \sqrt{n \ln n}]$ (Lemma~\ref{general_round}). Moving in or out at most $O( \sqrt{n \ln n})$ vertices, as $g=O(1)$, reduces the number of edges included in the subgraph by at most $n \sqrt{n \ln n}=o(n^2)$.

There is no known constant approximation poly-time algorithm for the problem. There is a deterministic greedy algorithm that achieves a ratio of $O(k/n)$~\cite{asahiro}, which is equal to $O(\gamma)$ in our case (we assume in this work that $\gamma$ is a constant). In~\cite{subgraph}, they give a randomized algorithm with approximation ratio at least $k/n \ge \gamma$. Finally, we can also use the original PTAS of~\cite{dense} for a not too small $\epsilon>0$.

\end{document}